\documentclass[review,  12pt,3p,authoryear,nopreprintline]{elsarticle}

\usepackage{mathtools,bm,amssymb,amsthm}
\usepackage[hidelinks]{hyperref}
\usepackage{wrapfig}
\usepackage{subcaption}
\usepackage{centernot,cancel}
\usepackage[inline]{enumitem}

\newcommand{\secref}[1]{Section~{\ref{#1}}}
\newcommand{\appref}[1]{{\ref{#1}}}
\newcommand{\figref}[1]{Figure~\ref{#1}}

\newcommand{\ddn}[2]{\frac{{d}^{#1}}{{d}x^{#1}} #2}
\newcommand{\ddna}[2]{\left.\ddn{#1}{#2}\right|_{x=a}}
\newcommand{\dd}[1]{\frac{d}{d{#1}}}

\newcommand{\ddy}{\dd{y}}
\newcommand{\bmtheta}{{\bm{\theta}}}

\newtheorem{thm}{Theorem}

\newcommand{\thmref}[1]{Theorem~\ref{#1}}

\newcommand{\lit}{LIT}

\newtheorem{corollary}{Corollary}
\newcommand{\qinfty}{\xrightarrow{q\to\infty}}
\newcommand{\kinfty}{\xrightarrow{k\to\infty}}

\newcommand{\pwk}{\xrightarrow{pw}}

\newcommand{\M}{{M}}

\newcommand{\st}{~\text{s.t.}~}

\begin{document}
	
\begin{frontmatter}
\date{}
\title{ Model Selection and Parameter Inference through Constraints \emph{via} Sequences of Surrogate Smoothing Functions\tnoteref{t1}}
\tnotetext[t1]{This research was supported by NSERC Discovery Grant RGPIN-2018-06787}

\author[1]{Mateen R Shaikh\corref{cor1}}
\ead{mshaikh@tru.ca}

\address[1]{Department of Mathematics and Statistics, Thompson Rivers University, Kamloops BC, Canada}

\cortext[cor1]{Corresponding author}

\begin{abstract}
	Models with fewer parameters are often easier to interpret and more robust. Parsimony can be achieved through optimizing objectives like the AIC or BIC, which are functions of the the number of free parameters in the model. Optimizing this discrete objective is a challenge, often relying on discrete optimization. We construct smooth functions with optima that reach the same optima of these objectives but permit continuous rather than discrete optimization, relieving some selection burden. Proofs of convergence are provided and a novel method of clustering through explicit overparamterization shows promising results.
\end{abstract}

\begin{keyword}
	Parameter Constraints, Variable Selection, Dimension Reduction, Discrete Regularization
\end{keyword}

\end{frontmatter}

	
	\section{Introduction}
	
	Model and variable selection arises in a variety of contexts and are often computationally expensive. Consider the usual objective of minimizing the Akaike Information Criterion 
	\citep[$AIC=-2\ell +2p $;][]{akaike74} 
	or the Bayesian Information Criterion  
	\citep[$BIC = -2\ell + p\log n$;][]{schwartz78} 
	where $\ell$ is the log-likelihood, $n$ is the number of observations and $p$ is the number of free parameters in the model. Imposing different combinations of constraints on parameters reduces the free parameter count by varying amounts, yielding a family of models to select from. The optimal combination of constraints is often chosen discretely, and expensively, by separately fitting each combination. 

	A popular alternative to the AIC/BIC is cross-validation (CV), e.g.,  it is commonly used in code for LASSO \citep{tibshirani96} variants. CV repeatedly fits subsets of the data and use the remaining data, using prediction to quantify model fitness. This is also computationally expensive  but relies more on data than assumptions to balance model fit and complexity, which makes a return to the AIC/BIC attractive. These approaches are consonant. Previous work has shown that the AIC is asymptotically equivalent to leave-one-out CV \citep{stone77} through a Taylor expansion. The BIC shown to be asymptotically consistent and equivalent to leave-$b$-out CV where $b= n-n[\log n-1]^{-1}$ \citep{shao97}. A generalized form of the AIC/BIC for regression, $-2\ell + c_n p$, yields consistent estimators when $c_n{\xrightarrow{n\to \infty}}\infty$ and when $\frac{c_n}{n}{\xrightarrow{n\to \infty}} 0$  \citep{rao89} but is not necessarily equivalent to any form of CV.

	\subsection{Smoothing Discrete Objectives}\label{sec:contdiscreteobj}
	\noindent Rather than fitting many different models discretely, the LASSO  adds the continuous but non-differentiable $L_1$ penalty with estimates intentionally shrunk to 0, and hence biased. It zeroes coefficients when the optimum lies on the penalty's cusp. The Smoothly Clipped Absolute Deviation penalty \citep{fan2001} removes some bias with a piecewise, non-differentiable penalty. These approaches contrast other purely continuous approaches requiring separate posthoc inference, which can induce additional inferential problems such as multiple testing.
	
	\citet{kreimer_rubinstein_1992} used smoothing on non-differentiable objectives. \citet{nesterov2005smooth} provided proofs with uniform convergence in optimization when smoothing over regions of convexity and \citet{chen_2012} provide properties of the smoothing functions themselves, including the case where the solution space is known to be convex but the objective is continuous and non-convex. None of these apply to discontinuous functions such as AIC/BIC. Our proposed methodology will fill this gap by applying to such objectives.
	
	To compute, the approach will invoke The Lagrange Inversion Theorem \cite[\lit;][]{lagrange88, whittaker2020course}, which has been applied in Statistics for saddlepoint approximations\citep{daniels1954,jensen1995}, quantile estimation \citep{cornish1937,fisher1960} and distributions of order statistics \citep{hall1992}. 
	
	The next Section provides the proof and explains how discontinuities from parameter counts may be approximated. \secref{sec:examples} provides two illustrative examples. The first provides a visualization of the surrogate objectives. The second discusses an application that intentionally overparamterizes the problem with successful results in a novel approach to model-based clustering. \secref{sec:discussion} discusses ongoing and future work with the proposed methodology.

	\section {Methodology}
	
	We propose constructing well-behaved surrogate functions that provides function values that approach those of a specified, poorly-behaved objective. The surrogate is then deterministically perturbed to be closer to the true objective, using the optimum from the previous iteration as a seed to estimate the updated surrogate's optimum. An algorithmic parameter we call $k$ indicates how close the function is to the true objective, with $k\to\infty$ yielding the true objective. The reciprocal of $k$ serves a similar function to temperature in simulated annealing. At each iteration, complete optimization of the surrogate is not necessary. The iterate must simply continue to fall anywhere within a narrowing radius of convergence to weather the increasingly turbulent surrogate functions. We use analysis terminology including  smooth (specifically $C^{m}$), compact, interior, and pointwise convergence, denoted $\pwk$.

	\subsection{Derivations}
	\begin{thm}[Lagrange Inversion Theorem; \citealt{lagrange88}]\label{lit}
		Let $y=f(x)\in C^{\infty}(\mathcal{X})$ have a root at $x=A$. Let $a\in \mathcal X \st f'(a)\neq 0$. Then, $A$ may be estimated by
		\begin{align*}			A &=  a+\sum_{i=1}^{\infty} \left(-\frac{1}{i!}\frac{f(a)}{f_1(a)}\right)^i \ddna{i}{\frac{1}{f_1(x)}}\\
		\quad &= a-\left(\frac{f(a)}{f_1(a)}\right)+ \frac{1}{2!}f_2\left(\frac{f(a)}{f_1(a)}\right)^2-\frac{1}{3!}(3f_2(a)^2-f_1f_3(a))\left(\frac{f(a)}{f_1(a)}\right)^3 +\dots&& \label{grunge}
		\end{align*}
		where $f_j(a)=  \ddna{i}{f(x)} $, the $j^\text{th}$ derivative of $f$ evaluated at $x=a$.
	\end{thm}
	\begin{proof}
		Let $g(y)=f^{-1}(y)$, the inverse of $f(x)$. The proof follows as a Taylor series of $g(y)$ about $y=f(a)$ recalling that $\left.\ddy g(y)\right|_{y=f(a)}=\frac{1}{f'(a)}$ and substituting $y=0$.
	\end{proof}
	 \appref{app:finite} provides a proof of a corollary that the series may be truncated for arbitrarily precise finite approximations to estimating roots. This is useful in the main contribution, discussed in the next theorem. 
Although the theorem's antecedent may test the reader's patience, it's restrictions are curated to still easily facilitate Statistical optimization tasks, as demonstrated in the next section. The reader may refresh on these definitions in \appref{app:definitions}. Below, we change derivative notation to $M'(\theta)$, $M'_k(\theta)$, etc.,  denoting the usual derivatives with respect to argument $\theta$. This highlights that the argument and points of evaluation are the focus in the theorem statement and  proof, rather than the variable of differentiation, which remains $\theta$  throughout.

\begin{thm}[Optima Convergence  \emph{via} Sequences of Smooth Surrogate Functions]\label{thm}
	
	Let $\M(\theta)$ be defined for all $\theta$ in compact support $\Theta$, and  let ${\theta}^*\in\Theta$ represent the ordinate of a minimum(maximum) of $\M(\theta)$.  For $k>0$, consider any $\M_k(\theta)\in C^m(\Theta) \st  \M_k(\theta)~\pwk~\M(\theta), \forall \theta\in\Theta$. Then,
	for every compact $\Theta_k\subseteq\Theta$ with the all three of the following properties:
	\begin{enumerate*}[label=(\roman*),ref=(\roman*)]
		\item $\theta^*\in \Theta_k$; \label{item:includesoptimum}
		\item $\forall \theta \in \Theta_k, M_k(\theta)$ is strictly convex (concave); and   \label{item:convex}
		\item $M_k(\theta)$ has min(max), say $\theta^*_k$ in the interior of $\Theta_k$: $\theta^*_k\in\Theta_k^\circ$, \label{item:singleoptima}
	\end{enumerate*}%
the following three consequents follow: 
		\begin{enumerate*}[label=(\alph*),ref=(\alph*)]
		\item  $M'_k(\theta_k^*)=0$; \label{item:derivativezero}
		\item $\theta_k^*$ may be obtained through continuous optimization requiring the objective is no more smooth than $C^{m}$\label{item:ctsoptimization}; and	
		\item  $\theta^*_{k}\kinfty \theta^*$. \label{item:optimumconverges}
	\end{enumerate*}	
\end{thm}
\begin{proof}
			The circumscribed wording of the theorem statement makes the proof simple.  Since $\Theta_k$ is also compact, $M_k\in C^m(\Theta) \Rightarrow M_k\in C^m(\Theta_k)$. Also note that because of convexity (property \ref{item:convex}) for $M_k$, and $\theta^*_k\in\Theta_k^\circ$ (property \ref{item:singleoptima}), $M'_k(\theta)=0$ iff $\theta=\theta_k^*$, 
			resulting in \ref{item:derivativezero}. The extreme value theorem guarantees the optimum exists and smoothness permits continuous optimization which may exploit $C^m$, satisfying consequent \ref{item:ctsoptimization}.
			
			We show consequent \ref{item:optimumconverges} through contradiction. Consider when the three properties are satisfied but $\theta_k^*\centernot{\pwk} \theta^*$. If $\exists \theta_0 \st \theta_k\pwk\theta_0\neq \theta^*$, then either the function obtains another min(max) in the interval, contradicting  \ref{item:singleoptima}; or $M'_k(\theta_0)=0$ as a saddlepoint or max(min), contradicting \ref{item:convex}; or $\theta_k$ cannot reach $\theta^*$ because it lies outside the interval, contradicting \ref{item:includesoptimum}; or  
			$M_k(\theta_k^*)\pwk m_0\neq M(\theta^*)$ for some $m_0$,
			violating the pointwise convergence assumption. 
%
%
\end{proof}
As mentioned above, this does \emph{not} require that limit $\M(\theta)$ is differentiable, or even continuous, like AIC/BIC. Uniform convergence, which makes the result obvious, would require that the objective is continuous if the terms of the sequence are continuous, hence disqualifying those useful Statistical objectives.
This makes \thmref{thm} useful in Statistical optimization. Also note that smoothness only relevant for continuous optimization, otherwise continuity ($C^0$) suffices.

\subsection{Implementation for Parsimony}

In a model considering up to free $q$ parameters, we define the number of selected non-zero parameters $p=\sum_{j=1}^q \mathbb{I}(\theta_j\neq0)=q-\sum_{j=1}^q\mathbb{I}(\theta_j=0)$. We make a sequence of surrogate functions for the latter indicator ${d}_k(\theta_j)$ that approaches $\mathbb{I}(\theta_j=0)$. 
The objective to optimize, say for $AIC$, is given by $AIC=-2\ell(\bmtheta)+2p$. Decomposing the indicator count, we obtain
	\begin{alignat}{4}
			AIC= & -2\ell\phantom{_j}(\bmtheta)&&+2\sum 1-{\mathbb{I}({\theta_j=0}}).\\
			\shortintertext{Employing the (see below) surrogate smoother $d_k(\theta_j)\pwk\mathbb{I}({\theta_j=0})$, we construct }
			AIC_k=&-2~\ell\phantom{_j}(\bmtheta) +2q&&-2\sum d_k(\theta_j)\label{eqn:aick}\\
			\frac{\partial}{\partial \theta_j}AIC_k=&-2S_j(\bmtheta)\phantom{+2q}&&-2\sum d_k'(\theta_j)\label{eqn:scorek}\\
			\frac{\partial^2}{\partial \theta_j^2}AIC_k=&-2S'_j(\bmtheta)\phantom{+2q}&&-2\sum d_k'(\theta_j)\label{eqn:scorekk}\\[-\baselineskip]
			\vdots\notag
	\end{alignat}
	where $S_j(\bm\theta)=\frac{\partial}{\partial\theta_j}\ell(\bmtheta)$, is the usual score function for $\theta_j$.
	
	First note we choose some $d_k(x)$ that approaches the delta function, such as  sech$(kx)$, $\text{exp}(-\frac{1}{2}kx^2)$, $[2(1+x)^2]^{-1}$ etc.. These functions are all similarly shaped and share the property that derivatives may be written as functions of the original, which aids computation. Secondly, note that we only provide univariate derivatives as higher order derivatives requiring matrices and tensors become burdensome so simple univariate approximations are considered. 
	To obtain the optimum of \eqref{eqn:aick}, we estimate the root of Equation \eqref{eqn:scorek}. Optima may be found through a descent-style root-solving algorithm such as the truncated polynomial of the \lit, or Fisher scoring by substituting $S'_j$ with either Fisher or Observed information of $\theta_j$.

	Following the fuse approach\citep{tibshirani05}, parameters may be penalized for  differing from each other rather than differing from a constant, inducing a parameter clustering problem in the broader learning problem. The surrogate for $p$ is $1+\sum_{j=2}^q d_k(\theta_{(j)}-\theta_{(j-1)})$ where $\theta_{(j)}$ denote ordering with index $j$. Surrogates of the free parameter count are shown below. 
	\begin{align}
		p   & = 1+\textstyle \sum_{j=2}^q \mathbb{I}(\theta_j-\theta_{j-1}=0)\\
		p_k & = 1+\textstyle \sum_{j=2}^q d_k(\theta_j-\theta_{j-1}=0) \label{eqn:pk} \\
		\frac{\partial}{\partial \theta_{(j)}}p_k & = 0+\begin{cases}
			\phantom{d'_k(\theta_{(p)}-\theta_{(p-1)})}-d'_k(\theta_{(2)}-\theta_{(1)}) & j=1\\
			d'_k(\theta_{(j)}-\theta_{(j-1)})-d'_k(\theta_{(j+1)}-\theta_{(j)}) & 2\leq j \leq {q-1}\\
			d'_k(\theta_{(q)}-\theta_{(q-1)}) & j=q
		\end{cases}
		\shortintertext{We slightly abuse notation by extending parameter indices to define $\theta_{(0)}$ and $\theta_{(p+1)}$, which are never variables of differentiation, to any constants, so that $d_k'$ exists to be multiplied by zero for the two boundary cases of $j=1$ and $j=q$. This simplifies the derivative of $p_k$ as}
			\frac{\partial}{\partial \theta_{(j)}}p_k &=\qquad \mathbb{I}(j\geq2)\quad d'_k(\theta_{(j)}-\theta_{(j-1)})\qquad -\qquad \mathbb{I}(j\leq q-1)\quad   d'_k(\theta_{(j+1)}-\theta_{(j)})
	\end{align}
	\vspace{-1em}
	with further derivatives simply following the same pattern.

	\section{Examples}\label{sec:examples}
	\subsection{AIC surface}

	Our first  example demonstrates the solution path incrementing $k$ in a toy data set on simple linear regression designed to illustrate that the least squares solution provides a minimum slightly inferior to the null model. We use $p_k=\text{sech}(kx)$ as surrogate in the count indicator. The surrogate AIC curves and path of the optima are shown in \figref{fig:aicreg}. For  $k>20$, each of the two optima are reachable given different starting values empoying the LIT. Importantly, there was no guarantee that the optimum for the small $k$ would follow a path that led to the global optimum of the AIC and selecting surrogates to manipulate the path is left for future work.

			\newcommand{\figwidth}{0.25\textwidth}
\setlength{\intextsep}{-0\baselineskip}%
\setlength{\columnsep}{10pt}
\begin{wrapfigure}{r}{\figwidth}
	\vspace{-3.5em}
	\includegraphics[width=\figwidth]{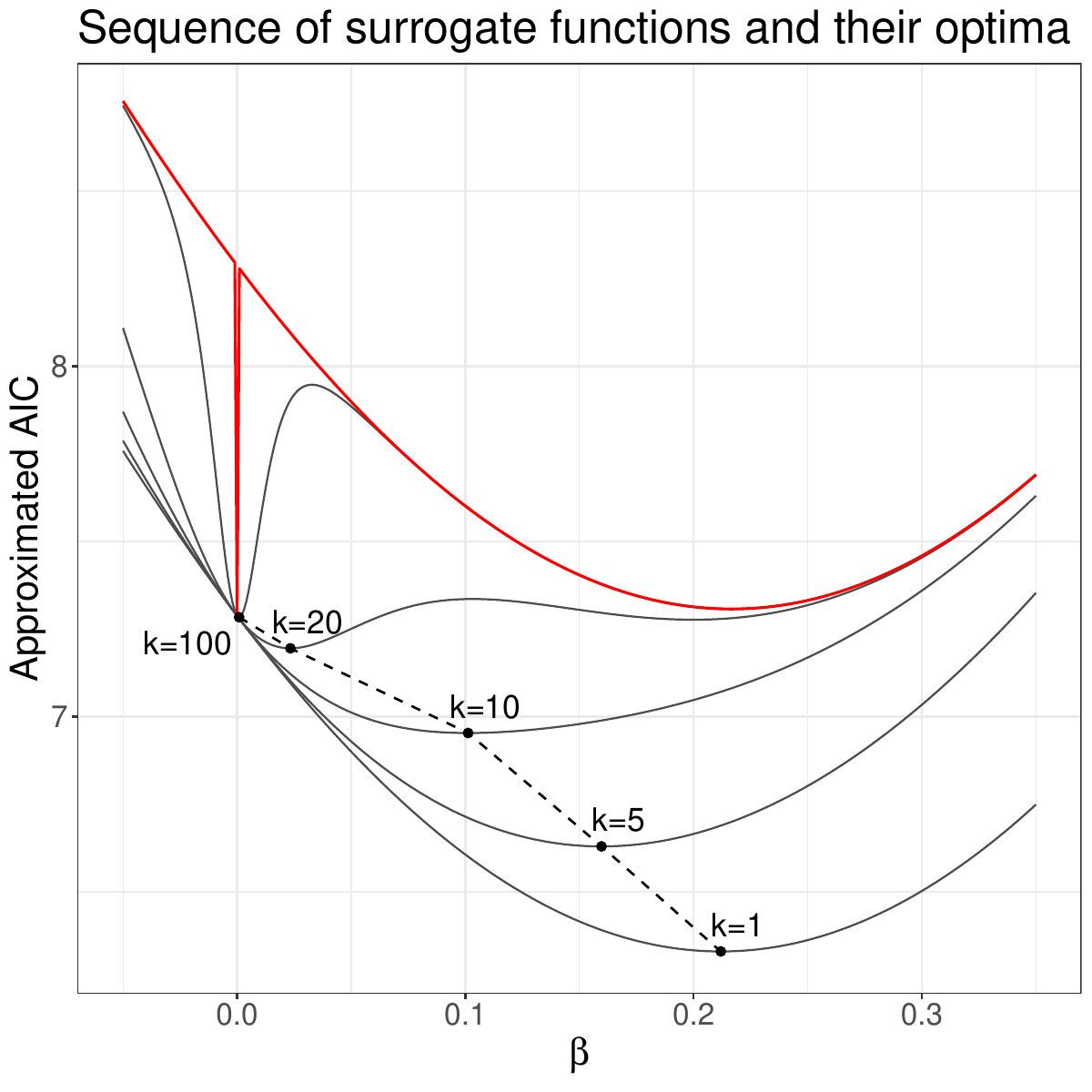}
	\caption{ The sequence of optima (blue) for surrogates (black) that approach the objective AIC (red).\label{fig:aicreg}}
\end{wrapfigure}

	\subsection{Parameter Constraints via Overparameterization}
	Our second example induces a novel method of model-based clustering. Families of finite mixture models are often used for this task but have two notable challenges. The first is that unconstrained high-dimensional models have a glut of parameters in multiplicity for each mixture component, which can lead to overparameterization. The second is that the number of components must be explicitly specified, making the ideal number a further discrete optimization task. Remarkably, neither of these are challenge when employing \thmref{thm}. We illustrate with the \texttt{faithful} \citep{azzalini1990} data set which provides bivariate observations of the duration of a geyser's eruption and the duration of the wait until its next eruption. Although the scatterplot suggest skew, we only consider the Gaussian to keep the example short. We also subset to every fifth observation to make observations distinguishable in the Figures.
	Rather than a bivariate Gaussian, we more easily illustrate clustering through separate of univariate Gaussians which provide similar results to a bivariate Gaussian with identity covariance. 
	Optimizing the surrogate $BIC$ again employing $d_k(x)=\text{sech}(kx)$ as the indicator's surrogate results in solution paths in $k$ for each of the mean parameters corresponding to the eruption and waiting time for each observation. These are illustrated in 
	\figref{fig:clusteringresults}.  As mentioned, the number of groups was not specified. Because the parameters fuse automatically, the fusions implicitly define univariate groups. Across all components, when the estimated mean vectors for observations were completely identical, the corresponding observations were declared to be in the same clustering group. A peculiar observation is apparent in \figref{fig:clustering}: one observation's mean eruption time identified with group 1 but its mean waiting time identified with group 2. This was indicated in the Figure as a ``Split'' decision. This illustrates a limitation  of applying the clustering as through independent univariate distributions. 

		\begin{figure}[!b]
	\vspace{-1em}
	\hspace{-2em}
	\begin{subfigure}{0.32\textwidth}
		\includegraphics[width=\textwidth]{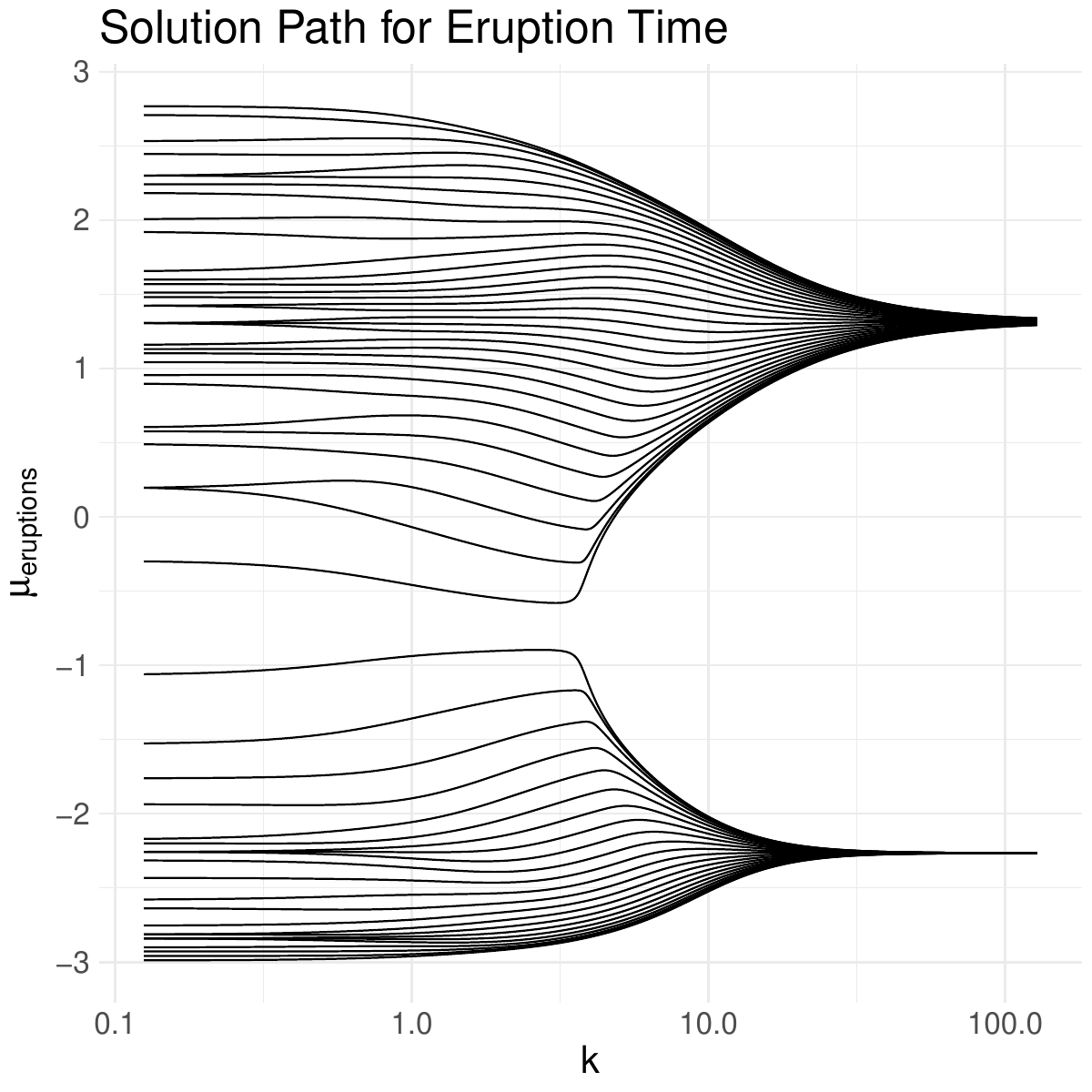}
		\caption{Solution paths of means corresponding to eruption times, over $k$.}
		\label{fig:eruptions}
	\end{subfigure}
	\begin{subfigure}{0.32\textwidth}
		\includegraphics[width=\textwidth]{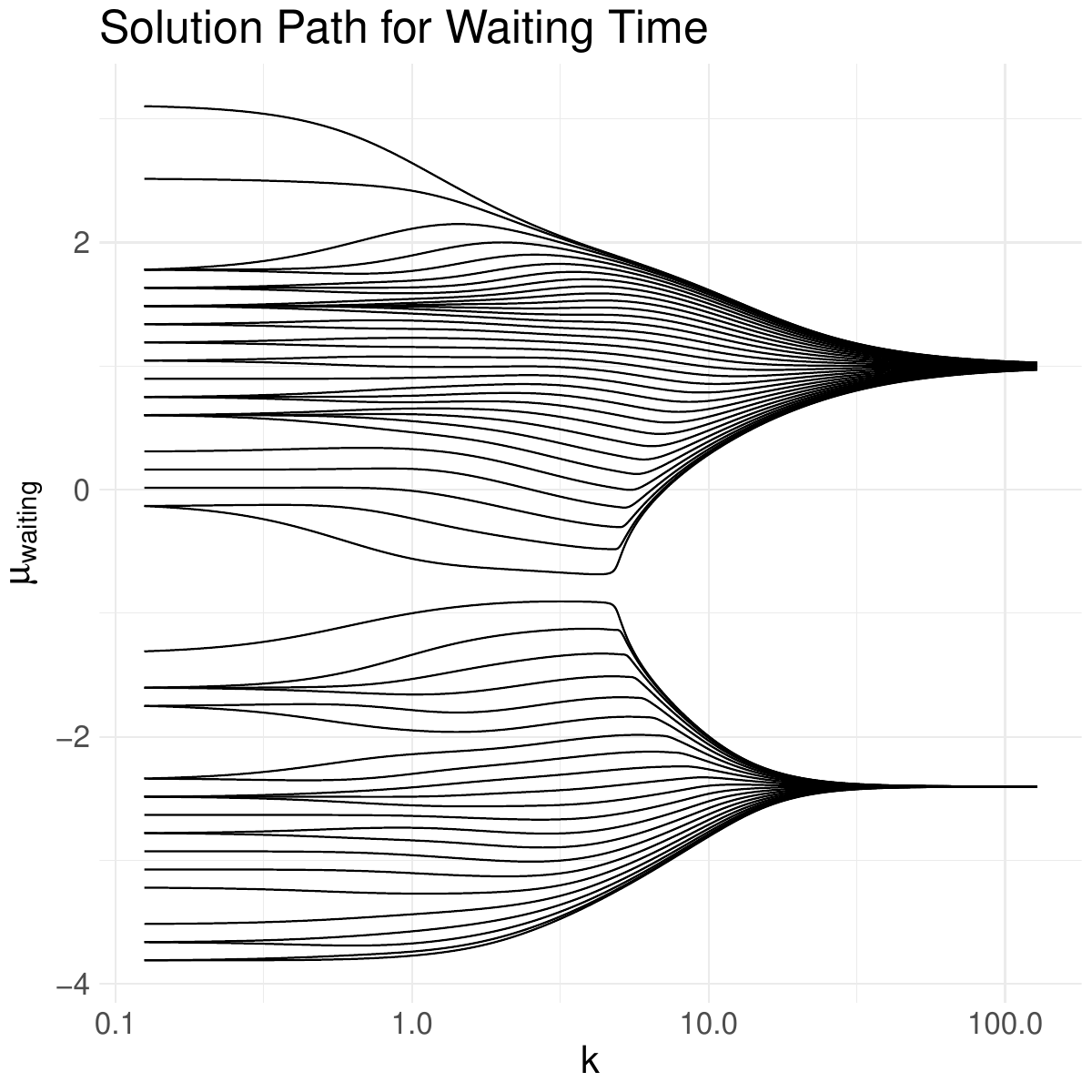}
		\caption{Solution paths of means corresponding to waiting times, over $k$.}
		\label{fig:waiting}
	\end{subfigure}
	\begin{subfigure}{0.42\textwidth}
		\includegraphics[width=\textwidth]{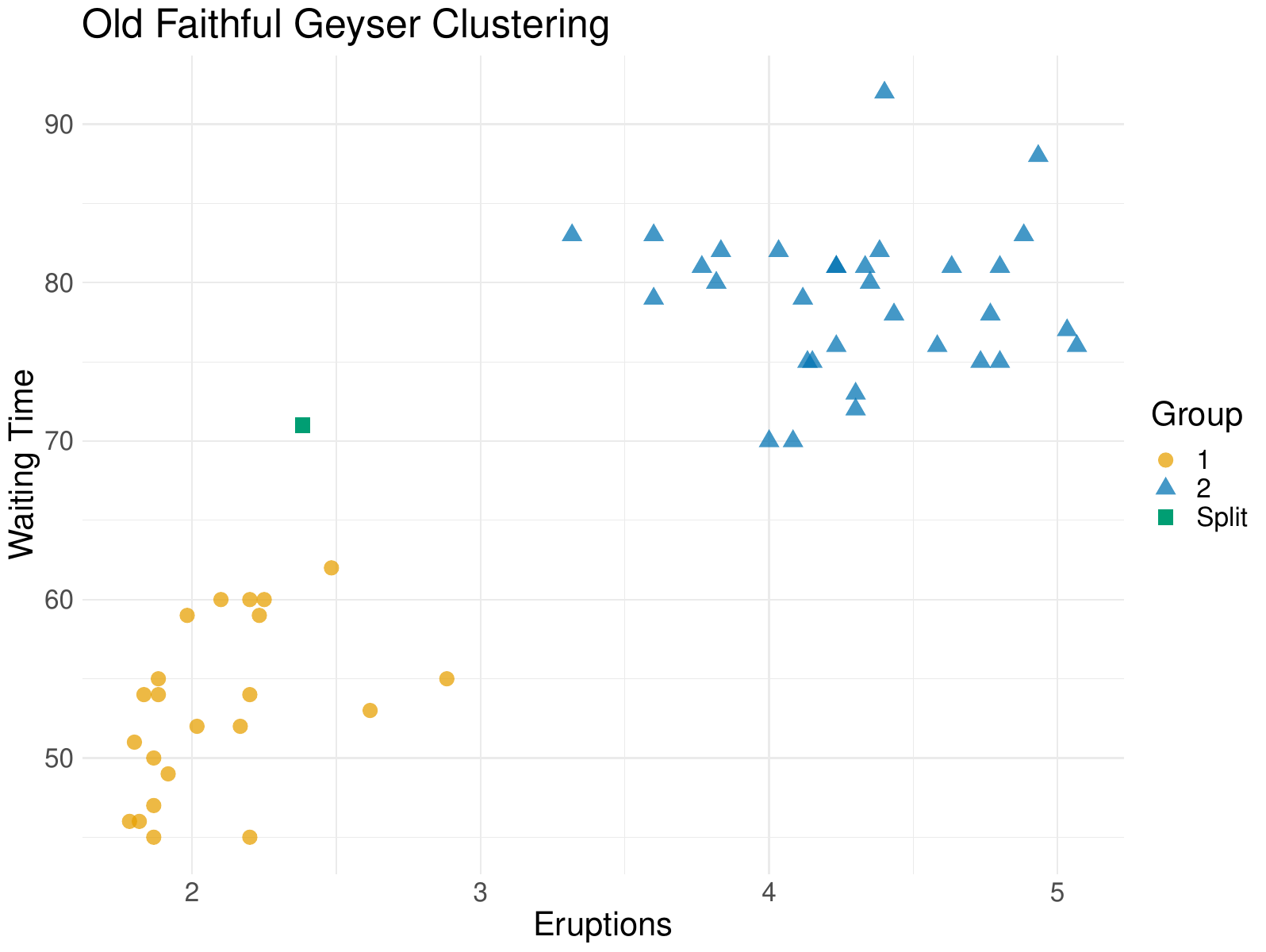}
		\caption{Clustering results. Groups are identified when all coordinate of the means are identical.}
		\label{fig:clustering}
	\end{subfigure}
	\caption{Results from employing sequences of surrogate functions to a subset of the faithful data.}
	\label{fig:clusteringresults}
\end{figure}

	\section{Discussion}\label{sec:discussion}

	We derived a method of optimizing discontinuous objectives such as the AIC/BIC through sequences of smooth surrogate functions, and proved its convergence. Application shows promise with some caveats spawning future work. One is that there is no guarantee that sequences converges to a global optimum but manipulating the surrogates may make it more likely.
	In higher-dimensional scenarios, simply ordering the scalar parameters and clustering them univariately is naive, though effective in these examples. Considering all pairwise constraints \citep{tibshirani2011}  can extend to multivariate parameter cases. This would create a surrogate adjacency matrix and corresponding surrogate Laplacian matrix, inducing a graph. The number of components in the graph represents the number of free parameters and this may be incorporated into the objective through (surrogate) eigenvalues of the Laplacian. Finally a novel clustering approach was induced by our method that performed well by explicitly overparameterizing. This obviates discretely fitting many different models with varying number of components or constraints. We may explore directly maximizing objectives that asymptotically match CV through this approach, as in a generalized Information Criterion. Finally, we plan to explore discontinuous posteriors by incorporating the work of \citet{cornish1937} who heavily use the \lit,  which may allow us to circumvent the expense of sampling in Bayesian analyses.

	
	\setlength{\bibsep}{0pt} 


	\bibliographystyle{elsarticle-harv}
	\bibliography{../../../../references} 
	\newpage


	\appendix
	\label{appendix}
	
	\section{Appendix}
	\subsection{Terminology}\label{app:definitions}
		The following standard definitions are found in common textbooks on real analysis.  
	\begin{itemize}
		\item A \textbf{compact} set is one that is closed and bounded.
		\item The \textbf{interior} of set $\mathcal X$ is the largest open set in $\mathcal{X}$. In the univariate case if  $\mathcal{X}=[a,b]$ the interior of $\mathcal{X}$ is $\mathcal{X}^\circ=(a,b)$.
		\item The scalars $z_k$ \textbf{converge} to $z$, denoted $z_k\kinfty z$,  if $\forall \epsilon>0, \exists N \st  \forall k>N$ it follows that$|z_k-z|<\epsilon$. We do very slightly modify the usual definition to not restrict $k$ and $N$ to be integers, but only require unambiguous ordering so the indices may not be integers but sequence terms remain unambiguously ordered. The bridge between the two definitions is an implied sequence indexed by $i\in \mathbb N$ where $i\mapsto k_i \mapsto z_{k_i}$ for some fixed sequence of $\{k_i\}\subset \mathbb R$. In our work the $i$ is implicit and we directly use $k=k_i\in \mathbb R$.

		\item For $k>0$, the functions $g_k(x)$  \textbf{converges pointwise} to $g(x)$ in $k$ on domain $\mathcal X$, denoted as $g_k(x)~\pwk~g(x)$,  if $\forall x\in \mathcal X, g_k(x)\kinfty g(x)$, or without relying on the usual scalar convergence definition, $	
		\forall \epsilon>0, \exists N \st \forall k>N $ it follows that $|g_k(x)-g(x)|<\epsilon$. 
		
		\item A function $g(x)$ on domain $\mathcal{X}$ is \textbf{smooth} (or $C^{m}$ smooth) on $\Theta$ if for $n \in \{0,1,\dots,m\}$ when  $\ddn{n}{g(x)}$ is continuous $\forall x\in \mathcal{X}$. Let the $n=0$ case refer to $g(x)$ itself. Note that $C^{n+1}\implies C^n$. The notation for this is $g(x)\in C^m(\mathcal X)$. 
	\end{itemize}
	\subsection{Lagrange Finite Approximation}\label{app:finite}
		\begin{corollary}[Finite Approximation]\label{cor:finite}
		Let $f(x)$ represent the same smooth function as in the 
		\lit with root at $x=A$.
		Let $h_q(0)$ represent the $m^{\text{th}}$ order Taylor expansion of $g(y)=f^{-1}(y)$ about $y=f(a)$  evaluated at $y=0$.Then $	\displaystyle \lim_{m\to\infty}h_q(0)=A$.
	\end{corollary}
	\begin{proof} The result is a direct application of the Taylor Remainder theorem. 
		Since, $h_q(0)= a+\sum_{i=1}^{q} \left(-\frac{1}{i!}\frac{f(a)}{f_1(a)}\right)^i \ddna{i}{\frac{1}{f_1(x)}}$ from from the \lit, The Taylor Remainder Theorem gives
		\begin{align*}
			R_q(y)&=o(|y-b|^q)= o(|y-f(a))|^q)\\
			\implies\quad  R_q(0)&=o(|f(a)|^q)\\
			\implies\quad \qquad  A &= g(0) = h_q(0) + o(|f(a)|^q) 
		\end{align*}
		As $f(a)$ is meant to be close to zero, an arbitrary constant $c_a$ may be applied to $f$ to ensure that $|c_af(a)|<1$. $|c_af(a)|^q \qinfty0$ which results in $|f(a)|^q\qinfty 0$. Thus $h_q(0)\qinfty A$ for all~$a$.
	\end{proof}
	
\end{document}